\documentclass[10pt,conference]{IEEEtran}
\usepackage{graphicx,subfigure}
\usepackage{amsfonts,amsmath,amssymb, mathrsfs}
\usepackage{epic,eepic,eepicemu}
\usepackage{epsf}
\usepackage{epsfig}
\usepackage{graphics}
\usepackage{psfrag}

\newtheorem{theorem}{Theorem}

\newtheorem{claim}{Claim}

\usepackage{tikz}
\usetikzlibrary{arrows,shapes,backgrounds,snakes}

\tikzstyle{dot}=[circle,draw=gray!90,fill=gray!20,thick,inner
 sep=1pt,minimum size=16pt]
\date{}

\usepackage{array}

\begin{document}
\title{Transmission of non-linear binary input functions over a CDMA System}
\vspace{-3.5cm}
\author{\IEEEauthorblockN{Elaheh Mohammadi}
\IEEEauthorblockA{Dept. of EE\\
Amirkabir University of Technology\\
Tehran, Iran}
\and
\IEEEauthorblockN{Amin Gohari}
\IEEEauthorblockA{Dept. of EE\\
Sharif University of Technology\\
Tehran, Iran}
\and
\IEEEauthorblockN{Hassan Aghaeinia}
\IEEEauthorblockA{Dept. of EE\\
Amirkabir University of Technology\\
Tehran, Iran}
}

\maketitle

\begin{abstract}We study the problem of transmission of binary input non-linear functions over a network of mobiles based on CDMA. Motivation for this study comes from the application of using cheap measurement devices installed on personal cell-phones to monitor environmental parameters such as air pollution, temperature and noise level. Our model resembles the MAC model of Nazer and Gastpar except that the encoders are restricted to be CDMA encoders. Unlike the work of Nazer and Gastpar whose main attention is transmission of linear functions, we deal with non-linear functions with binary inputs. A main contribution of this paper is a lower bound on the computational capacity for this problem. While in the traditional CDMA system the signature matrix of the CDMA system preferably has independent rows, in our setup the signature matrix of the CDMA system is viewed as the parity check matrix of a linear code, reflecting our treatment of the interference.
\end{abstract}

\section{Introduction}
The problem of decoding functions of sources rather than the sources themselves in a Multiple Access Channel (MAC) has been studied in several works (see for instance \cite{NazerGastpar} and its follow up works, also \cite{KornerMarton} and \cite{Vishwanath}). It has been shown that separation is not always optimal in such scenarios, even when the sources are independent \cite{NazerGastpar}. The intuitive reason for this is that the interference caused by other users could be exploited to compute a given function over the air, if the pattern of the interference matches the functions we want to compute.

All of the previous models for transmission of functions over MAC (that we have seen) do not impose any restrictions on the structure of the encoders, except perhaps on the input power. However some promising emerging applications may violate this assumption. We were motivated by one such application to impose a CDMA system as being part of the encoders.

The application is monitoring the exposure of humans to environmental parameters such as air pollution, temperature, etc. since the authors are living in one of the world's most polluted cities. The traditional way of monitoring is to install measurement devices distributed over a given area. Suggestion has been made to install low cost measurement devices on personal cell-phones, e.g. see \cite{Economist}. Although the focus of this paper is not the application, but we would like to mention a motivation for this application since it may be new (we have not seen it in the literature). Suppose we are interested in the collective exposure of residents to air pollution (not just personal exposures or the general pollution maps). To find the answer, it is not sufficient to have a pollution map, but also the population density at the polluted areas at various times in a day. Let us take the average of the measurements by the mobile sensors that are being carried by the residents as they move in the city. There will be just more cell-phones in populated areas and we can simultaneously take into consideration both the pollution and the population density.

Note that when the mobile system in the application is employing the CDMA system, it is preferable to use the same architecture to transmit functions of measurements by the cell-phones. Thus, we are considering the problem of function transmission over a network of mobiles based on CDMA.

\begin{figure}
\centering
\includegraphics[width=90mm]{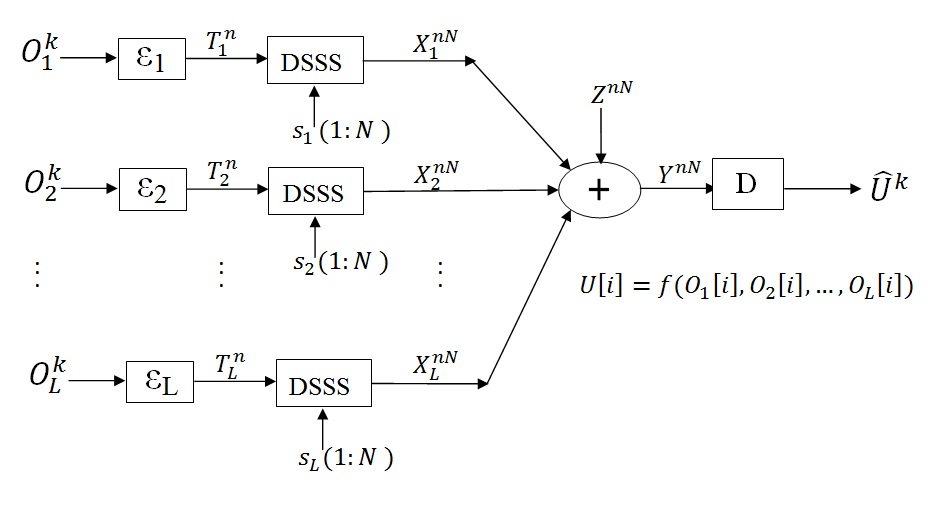}
\vspace{-1cm}
\caption{The communication model. The observations $O_1, \cdots, O_L$ are passed through encoders (resulting in $T_i$'s), multiplied by signatures (the $s_i$'s) of a CDMA systems and transmitted over the air (a Gaussian MAC channel). The receiver gets sum of the transmitted signals plus noise (i.e. $Y$), and wants to recover a function $f(O_1, \cdots, O_L)$ of the observations (denoted by $U$). In the general model we also allow for $b$ distinct functions instead of just one.}\label{fig:Setup}
\vspace{-0.5cm}
\end{figure}

We argue that finding the \emph{optimal} scheme for transmission of functions can be studied from two different criteria: i.e. maximizing privacy and minimizing transmission rates. Transmission of the whole data (rather than a function of it) may not only be bad in terms of transmission rates, but also it may compromise the privacy of cell-phone users. Therefore we can either maximize privacy or minimize transmission rates over all codes that allow reliable function computation. For the Korner-Marton problem \cite{KornerMarton} the two criteria yield exactly the same answer\footnote{This is because the equivocation rate for each user is bounded from above by the entropy of their source conditioned on the function to be computed, which is exactly what the Korner-Marton scheme achieves.}; but we believe in general they may be different. Nonetheless in this study we follow the traditional approach of minimizing the transmission rates.

Our model is shown in Fig. \ref{fig:Setup} which is similar to the one considered in \cite{NazerGastpar}, except for addition of the signature matrices. It is discussed rigorously in section \ref{Section:FormalDefinition}. But before that in section \ref{Section:Motivation}, we intuitively discuss our interpretation of the signature matrix of the CDMA system as the parity check matrix of a linear code, demonstrating our treatment of interference. \vspace{-0.01cm} Section \ref{Section:NazerGastpar} discusses the lower bound of \cite{NazerGastpar} applied to transmission of non-linear functions in our setting. Section \ref{Section:MainResults} contains our main result, providing a lower bound on the computational capacity for our problem. The bound is expressed in terms of the answer to another problem that we introduce, i.e. the problem of Slepian Wolf with the same compression matrices (discussed in section \ref{Section:Extension}). We believe the latter problem can itself be of independent interest. We will not be discussing any upper bounds, but one can derive an upper bound using the ideas in \cite{NazerGastpar} by merging all the transmitters into one node (the same technique used in some versions of the cut-set bound).

\section{Signature Matrix as a Parity Check Matrix }
\label{Section:Motivation}
In this section we discuss our use of the signature matrix as a parity check matrix at a very simplistic level to convey the basic intuitions. Let us assume that we have only three cell-phones. These cell-phones are observing binary random variable $O_1$, $O_2$ and $O_3$ respectively. The goal of the base station is to recover a boolean function of $O_1$, $O_2$ and $O_3$. Let us assume that the cell-phones directly insert their uncoded bits into a CDMA system with a given signature matrix. For instance, if the signature matrix is
\begin{eqnarray}\label{eqn:matrixH}\left(
  \begin{array}{ccc}
    1 & 1 & 0 \\
    1 & 0 & 1\\
  \end{array}
\right)\end{eqnarray}
the signature of the first, second and the third cell-phones would be the vectors $\textbf{s}_1=(1,1)^t$, $\textbf{s}_2=(1,0)^t$ and $\textbf{s}_3=(0,1)^t$ respectively. Assuming tight power control, the receiver gets the vector $O_1\textbf{s}_1+O_2\textbf{s}_2+O_3\textbf{s}_3$ plus some noise. Let us assume that there is no noise for now. In this case, the receiver gets two symbols, the first of which is $Y_1=O_1+O_2$ and the second one is $Y_2=O_1+O_3$. Note that the summation here is real addition in $\mathbb{R}$, and not in the field $\mathbb{F}_2$. Because $O_i$ takes values in $\{0,1\}$, $Y_1$ and $Y_2$ will be numbers in the set $\{0,1,2\}$. If $Y_1=0$, we can conclude that $O_1=O_2=0$. The value of $O_2$ would then specify $O_3$. Similarly, when $(Y_1, Y_2)=(2,1)$ we can figure out $O_1$, $O_2$ and $O_3$ exactly. However, when $(Y_1,Y_2)=(1,1)$, there are two possibilities: $(O_1, O_2, O_3)$ can be $(1,0,0)$ or $(0,1,1)$. If one were to compute a function $f(O_1, O_2, O_3)$ at the receiver, the necessary and sufficient condition for doing so would be that $f(1,0,0)=f(0,1,1)$. Note that this implies that among $2^{2^3}$ plausible boolean functions, half of them are computable with the given signature matrix. Now, observe that if we interpret the signature matrix given in equation \eqref{eqn:matrixH} as a parity check matrix, the codewords would be the triples $(0,0,0)$ and $(1,1,1)$. This implies that the triples $(1,0,0)$ and $(0,1,1)$ form a coset for this codebook, because their mod-$2$ sum is a codeword. The constraint $f(1, 0, 0) = f(0, 1, 1)$ says that $f$ has to be constant over this coset.

The above simple example can be extended to more general setups. It turns out that if we interpret the signature matrix as a parity check matrix, and take a function $f$ that is equal to a constant over any coset of the parity check matrix\footnote{Of course the constant may depend on the coset.}, we will be able to perfectly recover $f$ when the channel is noiseless. When the channel is noisy, one can overcome noise via pre-coding; this is explained formally in Sec. \ref{Section:MainResults}.

\section{The Communication Model}
\label{Section:FormalDefinition}
\begin{table}
\caption{Notation} \centering
\begin{tabular}{|c|c|}
    \hline
Variable &  Description  \\
\hline
    \footnotesize{$O_i (1\leq i\leq L)$}& \footnotesize{Observations by the nodes.} \\
\hline
\footnotesize{$U_i (1\leq i\leq b)$} & \footnotesize{Functions to be computed.}\\
\hline
\footnotesize{$T_i$} & \footnotesize{Output of the $i^{th}$ cell-phone}\\& \footnotesize{to be multiplied by the signature.}\\
\hline
\footnotesize{$\textbf{s}_i$} & \footnotesize{The signature of the $i^{th}$ cell-phone.}\\
\footnotesize{$N$} & \footnotesize{Length of the signatures.}\\
\hline
\footnotesize{$V_i (1\leq i\leq r)$} & \footnotesize{In most places $\sum_{j=1:L}h_i[j]O_j$ (modulo $2$).}\\
\hline
    \end{tabular}
\vspace{-0.45cm}\label{table}
\end{table}

In this section we define the communication model for our problem. Some of the notation we encounter as we go along the paper are summarized in Table \ref{table}. For a r.v. $T$ we use $T^n$ as a shorthand for the sequence $(T[1], T[2], \dots, T[n])$.

Assume that there are $L$ cell-phones. Let us denote the observation of the $i$-th cell-phone by r.v. $O_i$ taking values in the discrete set $\mathcal{O}_i$. R.Vs $O_1, O_2, \dots, O_L$ are jointly distributed according to a given $p_{O_1, O_2, \dots, O_L}(o_1, o_2, \dots o_L)$. We assume that the $L$ cell-phones are observing \emph{i.i.d.} copies $O_1, O_2, \dots, O_L$. The goal of the cell-phones is to enable the base station to recover i.i.d.\ copies of $b$ functions of the observations which we denote by $U_i$ ($1\leq i\leq b$), $U_i=f_i(O_1, O_2, \dots, O_L)$.

\emph{Definition of a code:} (see Fig. \ref{fig:Setup}) A code consists of
\begin{enumerate}
  \item An encoder for each cell-phone, mapping $O_i^k$ ($k$ i.i.d.\ copies of $O_i$) into a sequence of $n$ bits (denoted by $T_i^n$),
  \item A binary signature $\textbf{s}_i$ of length $N$ for the $i$-th cell-phone,
  \item One decoder at the receiver.
\end{enumerate}
\par
The actual signals transmitted over the air are $X_i^{nN}$ for $i\in[1:L]$, that are formed by multiplying each bit of encoder's output $T_i^n$ into the signature $\textbf{s}_i$. Note that the length of $T_i^n$ is $n$, and the length of $\textbf{s}_i$ is $N$. Since each bit of $T_i^n$ is multiplied by the whole sequence $\textbf{s}_i$ in the CDMA system, the output will be a binary string of length $nN$, denoted by $X_i^{nN}$.

Assuming a CDMA power control, the transmitted $X_i^{nN}$ goes through a Gaussian MAC channel, and the receiver gets $Y^{nN}$ where $Y[i]=\sum_{j=1}^LX_j[i]+Z[i], i\in[1:nN]$ for a Gaussian noise sequence $Z^{nN}$. The receiver (base station) takes the output sequences $Y^{nN}$ and passes it through a decoder to reconstruct the $b$ functions $\widehat{U}_i^k$ ($i\in[1:b]$). The probability of error of the code is taken to be the probability that $\widehat{U}_i^k\neq U_i^k$ for some $i\in[1:b]$. The rate of the code is taken to be $R=\frac{k}{nN}$. To impose a power constraint on the users, we assume that that $T_j[i]$ is taking values in $\{0,1\}$,\footnote{There will be a power gain by subtracting the mean of $T_j[i]$ from it, because that would reduce the variance of the transmitted signal. This would convert $T_j[i]$ into a $\pm$ bipolar signal. However, use of $0/1$ signals makes the exposition of the paper more appealing. Furthermore random variables $T_j[i]$ that we will end up using will have a uniform distribution over $\{0,1\}$ and can be adjusted to bipolar signals at the very last stage to decrease the average power consumption, while leaving the arguments unchanged.} and the variance of $Z[i]$ is $\sigma^2$.
%
%

\emph{Computational Capacity:} Given a signature length $N$, A communication rate $R_N$ is said to be achievable if there is a sequence of codes, $\mathcal{C}_n$ for $n\in \mathbb{N}$, all having signatures of length $N$, such that $\lim_{n\rightarrow\infty}P_e(\mathcal{C}_n)=0$ and $\lim_{n\rightarrow\infty}R(\mathcal{C}_n)=R_N$ where $P_e(\mathcal{C}_n)$ and $R(\mathcal{C}_n)$ are the probability of error and rate of the code respectively. The computational capacity for a signature length $N$, $C_N$, is taken to be the supremum of the set of achievable rates for that signature length $N$.

\section{Slepian-Wolf With The Same Compression Matrices}
\label{Section:Extension}
Before discussing our main result, we need to introduce the problem of Slepian-Wolf with the same compression matrices. We believe this problem can itself be of independent interest.

We first begin with the problem in a special case. Suppose we have three correlated binary sources $V_1$, $V_2$ and $V_3$ jointly distributed according to $p(v_1,v_2,v_3)$. I.i.d. copies of these three sources are observed by three parties, who want to communicate these i.i.d.\ copies to a fourth party, Alice, using noiseless links of rates $R_1$, $R_2$, and $R_3$. We are interested in the case of $R_1=R_2=R_3=R$. The minimum possible value of $R$ will be the minimum value of $R$ such that $(R,R,R)$ is in the Slepian-Wolf region. We call this $R_{SW}$. We know that for any $R>R_{SW}$ we can achieve the rate triple $(R,R,R)$ using linear codes: there are matrices $B_1$, $B_2$ and $B_3$ (of size $nR\times n$) where the three parties can use and send $B_1V_1^n$, $B_2V_2^n$ and $B_3V_3^n$ where $V_i^n$ is a column vector consisting of $V_i[j]$ for $j\in[1:n]$. The multiplication is in the field $\mathbb{F}_2$.

 Now, what if we are interested to find a \emph{single matrix} $B$, such that having $BV_1^n$, $BV_2^n$ and $BV_3^n$ we can recover ($V_1^n$, $V_2^n$ and $V_3^n$)? The three parties are sending at rates $R_1=R_2=R_2=R$ to Alice using the \emph{same} compression matrix $B$. We denote the minimum value of $R$ in this case by $R^{s.c.}_{SW}$. Clearly $R^{s.c.}_{SW}$ is larger than or equal to $R_{SW}$ (defined in the previous paragraph), because more restrictions are imposed on the definition of $R^{s.c.}_{SW}$. But is $R^{s.c.}_{SW}$ always \emph{equal} to $R_{SW}$? We show in Claim \ref{thm:2} that this is not true. The definition of $R^{s.c.}_{SW}$ can be extended to more than three parties in the natural way.

Use of the same matrix $B$ to compress correlated data (or Slepian-Wolf with the same compression matrices) arises naturally in our problem. It is also related to the ``syndrome technique" whereby a single code based is constructed for distributed compression (see for instance \cite{PradhanRamchandran}\cite{ZarasoaRoumyGuillemot}\cite{GehrigDragotti}). And after all, it is interesting to find the best compression rate one can achieve if a universal compression code is used by all nodes in a distributed source coding problem.

We do not know the exact value of $R^{s.c.}_{SW}$, but prove a few results about it.
\begin{claim}\label{thm:1} Let us assume we have only two binary r.v.'s $V_1$ and $V_2$. Let $K=V_1+V_2~(\mbox{mod }2)$. Then $R^{s.c.}_{SW}$ for transmission of these two r.v's is less than or equal to $\max(H(K), H(V_1|K))$.
\end{claim}
\begin{proof} Let $R=\max(H(K), H(V_1|K))$. Here is the sketch of the proof: let us generate the coordinates of the common compression matrix $B$ (of size $nR\times n$) uniformly and randomly from $\{0,1\}$. Then having $BV_1^n$ and $BV_2^n$, we can add them modulo two to get $B(V_1^n+V_2^n)=BK^n$. Since $R\geq H(K)$, $B$ is a good source code for recovering $K^n$ with high probability. Hence we can decode $K^n$ first. The Slepian-Wolf rate for recovering $V_1^n$ with $K^n$ serving as a side information is $H(V_1|K)$. Since $R\geq H(V_1|K)$, $B$ is a good SW code with high probability. Therefore we can find $V_1^n$. Having $V_1^n$ and $K^n$, we can also recover $V_2^n$.
\end{proof}

\begin{claim}\label{thm:2} There exists $V_1, \cdots, V_r$ such that the value of $R^{s.c.}_{SW}$ is strictly larger than $R_{SW}$. Next, for any $p(v_1, \cdots, v_r)$,  $R^{s.c.}_{SW}$ is less than or equal to $\min(rR_{SW}, \max_iH(V_i))$.
%
\end{claim}
\begin{proof} Let $V_1=V_2=\cdots=V_r$. Then $R_{SW}$, i.e. the minimum value of $R$ such that $(R,R,R)$ is in the Slepian-Wolf region, is equal to $\frac{H(V_1)}{r}$. However, $R^{s.c.}_{SW}$ is equal to $H(V_1)$.

To show that $R^{s.c.}_{SW}\leq rR_{SW}$ always holds, we start from an arbitrary code for $R_{SW}$, and construct another code for $R^{s.c.}_{SW}$. Take an arbitrary code with compression matrices $B_1, B_2, ..., B_r$ all of size $nR\times n$. Let $B$ to be equal to $[B_1^t~B_2^t~\cdot~B_r^t]^t$ where $\small{t}$ is the transpose operation. 
 One can verify that matrix $B$ is a valid common compression matrix, and is achieving the rate $rR$ for the problem of $R^{s.c.}_{SW}$. Thus $R^{s.c.}_{SW}\leq rR_{SW}$. Note this upper bound on the ratio $\frac{R^{s.c.}_{SW}}{R_{SW}}$ cannot be made smaller than $r$ because of the example given at the beginning of this proof. \par To show the inequality $R^{s.c.}_{SW}\leq \max_iH(V_i)$, observe that a random compression matrix of size $ n[\max_i(H(V_i)+\epsilon]\times n$ allows for recovery of $V_i^n$ from $BV_i^n$ (for all $i\in [1:r]$) with the average probability of error converging to zero. Thus a particular instance should also work.
\end{proof}



\section{Main results}
\label{Section:MainResults}
In this section we state our main results. Proof is given in Section \ref{section:proof}.

Let $f_i(O_1, \dots,O_L)$ ($1\leq i\leq b$) be a set of functions satisfying the property that $$f_i(o_1, o_2, \dots, o_L)=f_i(o'_1, o'_2, \dots, o'_L),~~i\in[1:b],$$ for any two sequences $(o_1, o_2, \dots, o_L)$ and $(o'_1, o'_2, \dots, o'_L)$ belonging to the same coset of some parity check matrix $H$.\footnote{Vectors $(o_1, o_2, \dots, o_L)$ and $(o'_1, o'_2, \dots, o'_L)$ belong to the same coset if $H[o_1, o_2, \dots, o_L]^t=H[o'_1, o'_2, \dots, o'_L]^t$ where the product is in $\mathbb{F}_2$.} Without loss of generality we can assume that $H$ has distinct rows $h_1$, $h_2$, ..., $h_r$. Thus matrix $H$ is of size $r\times L$.
\begin{theorem}\label{thm:3} For any signature length $N>r$, the following rate is achievable
$$R_N=\frac{c}{N\cdot \mathfrak{R}}\geq \frac{c}{N\cdot \max_iH(\bigoplus^{mod~2}_{j=1:L}~h_i[j]O_j)},$$
where $\mathfrak{R}$ is $R^{s.c.}_{SW}$ for the choice of $V_i=\sum_{j=1:L}h_i[j]O_j$ (modulo $2$). The second lower bound comes from applying Claim \ref{thm:2}, and is an explicit lower bound expression.

The number $c$ is the capacity of a channel with input alphabet $\mathcal{W}=\{0, 1\}$ and output alphabet $[-\frac{1}{2},\frac{3}{2}]$ defined as follows: the output is formed by adding $W$ to a Gaussian noise with variance $\frac{\sigma^2}{\lfloor\frac{N}{r}\rfloor}$, and then taking it modulo $2$, meaning that we add an integer multiple of $2$ to it to make it fall into the interval $[-\frac{1}{2},\frac{3}{2})$. Note that because of the symmetry the capacity occurs at a uniform input distribution.
\end{theorem}

\subsection{Comparison with the computational capacity of \cite{NazerGastpar}}
\label{Section:NazerGastpar}
In this section we discuss how our lower bound extends the result of Nazer and Gastpar in \cite{NazerGastpar}. We find the set of functions where we can use the result of Nazer and Gastpar, and the lower bound it gives us.

Our formulation above is similar to the one given by Nazer and Gastpar \cite{NazerGastpar}, except that we have a signature matrix here. Nonetheless, if we fix the signature matrices, we can think of a virtual channel between the encoder and decoders that includes the signature matrix. The input to this virtual channel is $(T_1, T_2, ..., T_L)$ and the output is $Y(1:N)=\sum_{i=1}^LT_i\textbf{s}_i(1:N)+Z(1:N)$ where the noise vector $Z(1:N)$ has covariance matrix $\sigma^2I$. If we use the virtual channel $n$ times, we get an output vector of size $nN$ that we were denoting by $Y^{nN}$. \par In this case we can write down the lower bound given in \cite{NazerGastpar} when we have a linear function over a field. We are mainly concerned with functions with binary inputs. The only linear function on the field $\mathbb{F}_2$ is the $XOR$ function. So this already puts limitations on the lower bounds we can get by \cite{NazerGastpar}. When $U_i$ (for $1\leq i\leq b$) is the XOR of a subset of the observations $O_1, ..., O_L$, we get the following lower bound
\small $$\frac{I(\oplus_{i=1}^LT_i;Y(1:N))}{NH(U_1, U_2, \dots U_b)},$$\normalsize
where the factor $N$ in the denominator comes from our definition of rate. Because we are free to choose the signatures $\textbf{s}_1$, ..., $\textbf{s}_L$ we can take maximum of the above expression over all $\textbf{s}_1$, ..., $\textbf{s}_L$.
\small{$$\max_{\textbf{s}_1,\cdots, \textbf{s}_L}\frac{I(\oplus_{i=1}^LT_i;Y(1:N))}{NH(U_1, U_2, \dots U_b)}.$$}\normalsize
The above result works only when the $U_i$s are the $XOR$ functions of subsets of $O_1, ..., O_L$, and it involves a maximization problem that we found hard to do, even when we have linear functions on a field.

\begin{figure}
\centering
\includegraphics[width=57mm]{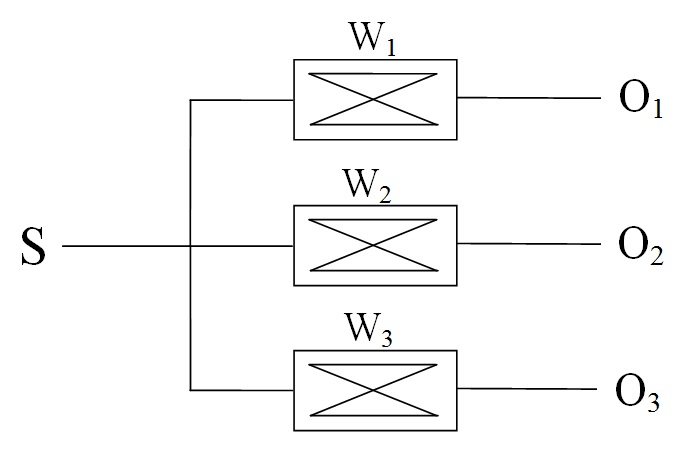}
\vspace{+0.43cm}
\caption{A model used for simulations. The observations $O_1, \cdots, O_3$ are assumed to be the result of a binary source $S$ passed through three independent BSC channels.}\label{fig:Model}
\end{figure}
\begin{figure}
\centering
\includegraphics[width=57mm]{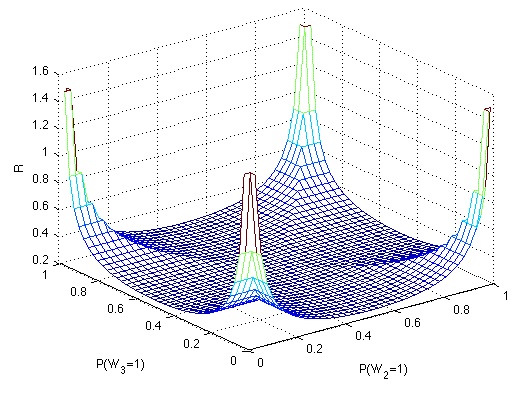}
\caption{The lower bound on the computational capacity for $O_1$, $O_2$ and $O_3$ of Fig. \ref{fig:Model}. The plot is in terms of $p(W_2 = 1)$ and $p(W_3 = 1)$ when
$p(W_1 = 1) = 0$. }\label{fig:Sim}
\vspace{-0.5cm}
\end{figure}

To compute arbitrary nonlinear boolean functions of the observations, Nazer and Gastpar suggest that we increase the field size and embed the non-linear function in a linear function defined on a larger space (see Theorem 2 of \cite{NazerGastpar}). Although this would not solve the maximization problem over the signatures $\textbf{s}_1$, ..., $\textbf{s}_L$ mentioned above, it will result in a lower bound for non-linear functions. In this paper we take an alternative approach of using several linear functions in the \emph{same} field using a \emph{particular construction} (rather than one single linear function over a larger field). In order to transmit several functions over a channel, \cite{NazerGastpar} uses a successive Slepian-Wolf type scheme. Our model allows us to do better than this. Through an appropriate choice of the signature matrix, we can run part of the transmission of the functions in parallel, getting an extra gain compared to the scheme considered by \cite{NazerGastpar}.

\section{Simulation}
Consider the boolean function $$f(O_1, O_2, O_3)=\overline{O_1}O_2O_3+O_1\overline{O_2}\overline{O_3} \mod~2.$$ This is not a linear function in the field $\mathbb{F}_2$. Let us assume that $N=2$. We can use the signature matrix given in equation \ref{eqn:matrixH} since $f$ is constant over all of its cosets. We have $V_1=O_1+O_2 (\mbox{mod }2)$ and $V_2=O_1+O_3 (\mbox{mod }2)$. Therefore
\begin{align*}&H(V_1)=h(p(O_1=O_2)),&\mbox{and}\\&H(V_2)=h(p(O_1=O_3)),&\end{align*}
where $h(\cdot)$ is the binary entropy function. The main theorem implies the following lower bound.
$$R=\frac{c}{2\cdot \max(h(p(O_1=O_2)), h(p(O_1=O_3)))}.$$
The value of $c\leq 1$ depends on $\sigma$. For the sake of illustration we assume that $\sigma$ is such that $c=0.5$.

Note that $R=\infty$ when \begin{align*}&p(O_1=O_2)\in\{0,1\},&~~~~\mbox{ and}\\&p(O_1=O_3)\in\{0,1\}.&\end{align*}
This is expected since in each of the four cases $f(O_1, O_2, O_3)$ is a constant. It would be interesting to understand the behavior of the lower bound when $p(O_1=O_2)$ and $p(O_1=O_3)$ are not exactly $\{0,1\}$, but rather in its vicinity. To study this, let us consider the model depicted in Fig. \ref{fig:Model} in which $O_1$, $O_2$ and $O_3$ are assumed to be the result of a random variable $B$ passing through three independent BSC channels, i.e.
\begin{align*}&O_1=S+W_1, O_2=S+W_2, O_3=S+W_3 \mod~2,\end{align*}
where $W_i$'s are binary random variables. $p(W_i=1)$ is the crossover probability of the $i^{th}$ channel. When $p(W_i=1)\in \{0,1\}$, the lower bound is $\infty$. Fig. \ref{fig:Sim} plots the lower bound $R$ in terms of $p(W_2=1)$ and $p(W_3=1)$ when $p(W_1=1)=0$.

\section{Proofs}
\label{section:proof}
\begin{proof}[Proof of Theorem \ref{thm:3}] We create the signature matrix of the CDMA by repeating the matrix $H$ to get a matrix of size $N\times L$. This means that each of the rows $h_1$, $h_2$, ..., $h_r$ would be repeated $\lfloor\frac{N}{r}\rfloor$ times; extra zeros are padded if $\frac{N}{r}$ is not an integer. At the receiver, we can look at the received $Y$'s corresponding to each of the $\lfloor\frac{N}{r}\rfloor$ repetitions and take their average. This would reduce the variance of noise for that transmission to $\sigma'^2=\frac{\sigma^2}{\lfloor\frac{N}{r}\rfloor}$. So, this would be as if the signature matrix is of size $r$ (instead of $N$) identical to $H$, and the noise variance is $\sigma'^2$ (instead of $\sigma^2$). We are going to continue assuming that the signature matrix and the parity check matrix are both $H$.

At time $i$, the cell-phones are sending $T_1[i], T_2[i], ..., T_L[i]$ respectively. The receiver gets $H\big[T_1[i], T_2[i], ..., T_L[i]\big]^t$ plus noise where the matrix multiplication here is in $\mathbb{R}$. To convert the matrix multiplication from $\mathbb{R}$ to that in $\mathbb{F}_2$, the receiver computes the modulo $2$ of each received number (as discussed in the statement of the theorem), mapping it to the interval $[-\frac{1}{2},\frac{3}{2})$. This would be as if $H\big[T_1[i], T_2[i], ..., T_L[i]\big]^t$ (matrix multiplication in $\mathbb{F}_2$) is transmitted but the noise added to this is no longer Gaussian; it is a Gaussian noise mod $2$. Number $c$ in the statement of the theorem is the capacity of this channel.

Having described the signature matrix, and the decoder's mod $2$ postprocessing of the signal, we now turn our attention to the encoders and the decoder. We can divide the rest of the proof into two parts. The first part is a general statement about recovery of the desired functions of the observations from $V_i$'s. This is used in the second part of the proof to design the encoders and the decoder.

\noindent(I) We first claim that given any values for $(o_1, o_2, \dots, o_L)$, knowing the values of $\sum_{j=1:L}h_i[j]o_j$ modulo two for $i\in[1:r]$ is sufficient to perfectly recover $f_i(o_1, o_2, \dots, o_L)$ ($1\leq i\leq b$). To see this note that having $r$ equations $\sum_{j=1:L}h_i[j]o_j$ (modulo two) for $i\in[1:r]$ is equivalent to having the product $H[o_1, o_2, \dots, o_L]^t$ in the matrix form; here the multiplication is in $\mathbb{F}_2$. Note that the number of equations is $r$ whereas the number of free variables is $L$, so it first seems that the decoder may not be able to figure out $[o_1, o_2, \dots, o_L]$. The decoder can list the set of all $[o'_1, o'_2, \dots, o'_L]$ such that $H[o'_1, o'_2, \dots, o'_L]^t$ (modulo two) is equal to the received $H[o_1, o_2, \dots, o_L]^t$ (modulo two). This would be the coset associated to $[o_1, o_2, \dots, o_L]$ for the parity check matrix $H$. Because $f_i$ maps all the sequences in a coset into the same number, namely $f_i(o'_1, o'_2, \dots, o'_L)$ are all equal, the receiver will be able to exactly recover $f_i(o_1, o_2, \dots, o_L)$.

\noindent(II) From the first part of the proof we can conclude that if we can reliably communicate i.i.d.\ copies of $V_i=\sum_{j=1:L}h_i[j]O_j$ (modulo two) to the receiver, it will be able to reliably recover i.i.d.\ copies of $f_i(O_1, O_2, \dots, O_L)$ ($1\leq i\leq b$). Therefore we have translated the original problem into that of communicating linear functions. If we think of the signature matrix as part of a virtual channel between the encoder and decoders, this virtual channel will be a set linear MACs (as defined by \cite{NazerGastpar}) in parallel because of the postprocessing at the receiver. Therefore our setting is not a special case of one considered by Theorem 1 of \cite{NazerGastpar} because the channel is not a \emph{single} linear MAC. Nonetheless we borrow ideas from \cite{NazerGastpar} to extend the proof of Theorem 1 of \cite{NazerGastpar}; this is not difficult given that the structure of the virtual channel and the linear functions to be computed (i.e. $V_i$s) are prepared to ``match".

It is possible to find a binary matrix $B$ of size $(k\mathfrak{R}+\epsilon)\times k$ for the i.i.d.\ copies of $(V_1,V_2,...,V_r)$ such one can recover i.i.d.\ copies of $V_1, V_2, \dots V_L$, namely $V_1^k, V_2^k, \dots, V_L^k$, from $B[V_1^k~V_2^k~\cdots~V_L^k]$.
%
within a probability of error $\epsilon$, where by $V_1^k$ we mean a column vector consisting of the $k$ i.i.d.\ copies of $V_1$. The multiplication between the column vector $V_i^k$ and $B$ is done in $\mathbb{F}_2$. Next we find a channel coding matrix $G$ of size $\frac{k\mathfrak{R}+\epsilon}{c-\epsilon}\times (k\mathfrak{R}+\epsilon)$ for communicating over a Gaussian channel with variance $\sigma'^2$. The $i^{th}$ cell-phone computes $GBO_i^k$. It sets this vector of size $n=\frac{k\mathfrak{R}+\epsilon}{c-\epsilon}$ to be $T_i^n$. At time $j$, the random variable $T_i[j]$ is multiplied by signature $\textbf{s}_i$. The receiver gets $\sum_jT_i[j]\textbf{s}_i$ plus a noise vector. This is equivalent with getting\small
\begin{align*}[T_1^n~T_2^n\cdots~T_L^n]H^t&=GB[O_1^k~O_2^k~\cdots~O_L^k]H^t\\&=GB[V_1^k~V_2^k~\cdots~V_L^k],\end{align*}\normalsize
%
plus noise. Since $G$ is a channel coding matrix, we can recover $B[V_1^k~V_2^k~\cdots~V_L^k]$ with high probability. From here we can recover $V_i^k$ because of the property of $B$ mentioned above. Thus, we have a good code. The rate of this code is \begin{eqnarray*}&\frac{k}{Nn}=\frac{k}{N\frac{k\mathfrak{R}+\epsilon}{c-\epsilon}}=
\frac{c-\epsilon}{N(\mathfrak{R}+\frac{\epsilon}{k})}.\end{eqnarray*}
Letting $\epsilon$ converge to zero, we get the desired result.
\end{proof}


\begin{thebibliography}{1}


\bibitem{Economist}
A. Gelman, ``Mobile phones: sensors and sensitivity", \emph{The Economist}, Technology Quarterly, 2009 (23).


\bibitem{NazerGastpar}
B. Nazer, M. Gastpar, ``Computation over Multiple-Access Channels," \emph{IEEE Trans. Inf. Theory}, 53 (10): 3498-3516,
2007.

\bibitem{KornerMarton}
 J. K\"{o}rner and K. Marton, ``How to encode the modulo-two sum of binary
sources," \emph{IEEE Trans. Inf. Theory}, 25 (2): 219–221, 1979.

\bibitem{Vishwanath}
R. Soundararajan and S. Vishwanath, \emph{Communicating Linear Functions of Correlated Gaussian Sources Over a MAC}, available at
$https://webspace.utexas.edu/rs6454/mac_linfun_journal.pdf$.

%
%
%
%


\bibitem{PradhanRamchandran}
S. S. Pradhan and K. Ramchandran, ``Distributed source coding using syndromes (DISCUS): design and
construction," \emph{Proc. DCC-1999, Data Compression Conf.}, pp. 158–167, 1999.

\bibitem{ZarasoaRoumyGuillemot}
V. Toto-Zarasoa, A. Roumy and C. Guillemot, ``Rate-adaptive codes for the entire Slepian-Wolf region and arbitrarily correlated sources" \emph{Proc. of the IEEE Int. Conf. on Acoustics, Speech and Signal Processing}, USA, pp. 2965-2968, 2008.

\bibitem{GehrigDragotti}
N. Gehrig and P. L. Dragotti, ``Symmetric and a-symmetric Slepian-Wolf codes with systematic and non-systematic linear codes," \emph{IEEE Commun.
Lett.}, 9(1):61–63, 2005.

\end{thebibliography}
\end{document}